\documentclass{article}
\usepackage{arxiv}
\usepackage[T1]{fontenc}
\usepackage[utf8]{inputenc}
\usepackage{geometry}
\usepackage{amsmath}
\usepackage{amsthm}
\usepackage{graphicx}
\usepackage{listings}
\usepackage{import}
\usepackage{tikz}
\usepackage{enumitem}
\usepackage{csquotes}
\usepackage[colorlinks=true, citecolor = blue]{hyperref}      
\usepackage{multirow, amssymb, amsmath, graphicx, arydshln, url}
\usepackage{algorithm, algorithmicx, algpseudocode}
\usepackage[T1]{fontenc}
\usepackage{natbib}

\usetikzlibrary{bayesnet}

\title{Modelling publication bias and \textit{p}-hacking}

\author{
  Jonas Moss \\
  Department of Mathematics, University of Oslo\\
  PB 1053, Blindern, NO-0316, Oslo, Norway \\
  \it{jonasmgj@math.uio.no} \\ 
     \And
  Riccardo De Bin \\
  Department of Mathematics, University of Oslo\\
  PB 1053, Blindern, NO-0316, Oslo, Norway \\
  \it{debin@math.uio.no} \\ 
}

\theoremstyle{plain}
\newtheorem{theorem}{\protect\theoremname}
\theoremstyle{definition}
\newtheorem{example}[theorem]{\protect\examplename}
\providecommand{\examplename}{Example}
\providecommand{\theoremname}{Theorem}

\def\bSig\mathbf{\Sigma}

\newtheorem{prop}[theorem]{Proposition}

\renewcommand{\sqrt}[1]{(#1)^{1/2}}

\providecommand{\tabularnewline}{\\}
\floatstyle{ruled}
\newfloat{algorithm}{tbp}{loa}
\providecommand{\algorithmname}{Algorithm}
\floatname{algorithm}{\protect\algorithmname}

\begin{document}
\maketitle
\begin{abstract}
Publication bias and \textit{p}-hacking are two well-known phenomena that strongly affect the scientific literature and cause severe problems in meta-analyses. Due to these phenomena, the assumptions of meta-analyses are seriously violated and the results of the studies cannot be trusted. While publication bias is almost perfectly captured by the weighting function selection model, \textit{p}-hacking is much harder to model and no definitive solution has been found yet. In this paper we propose to model both publication bias and \textit{p}-hacking with selection models. We derive some properties for these models, and we compare them formally and through simulations. Finally, two real data examples are used to show how the models work in practice.
\end{abstract}

\keywords{meta-analysis \and publication bias \and p-hacking}











\section{Introduction}

Meta-analysis the quantitative combination of information from different studies. Aggregating information from multiple studies brings about higher statistical power, higher accuracy in estimation and greater reproducibility. Unfortunately, it is not always possible to believe in the results of meta-analyses, as some model assumptions may be seriously violated. In particular, a meta-analysis must not be based on a biased selection of studies. Publication bias \citep{sterling1959publication} and \textit{p}-hacking \citep{simmons2011false} are the most common phenomena that violate these assumptions. 

Publication bias, also known as the file drawer problem, \citep[see, e.g.,][]{iyengar1988selection} denotes that phenomenon when a study with a smaller \textit{p}-value is more likely to be published than a study with a higher \textit{p}-value. Publication bias is a well known issue, and several approaches have been proposed to tackle it. Two famous examples are the trim-and-fill \citep{duval2000trim} and fail-safe $N$ \citep{becker2005failsafe} methods. Neither of them are \textit{bona fide} statistical models with likelihoods and properly motivated estimation strategies. From a statistical point of view, the most important class of models which are used to deal with publication bias are selection models. They were first studied by \citet{hedges1984estimation} for $F$-distributed variables with a cut-off at $0.05$, and extended to the setting of $t$-values by \citet{iyengar1988selection}. \citet{hedges1992modeling} proposed a random effects publication bias model with more than one cut-off, while \citet{citkowicz2017parsimonious} used beta distributed weights.

Publication bias is a well-known problem in several research areas, and therefore various approaches to solve the issue have been also proposed outside the statistical literature. Hailing from economics, the models PET and PET-PEESE \citep{stanley2014meta,stanley2017limitations} are two models based on linear regression and an approximation of the selection mechanism based on the inverse Mill's ratio. From psychology, the \textit{p}-curve of \citet{simonsohn2014p} is a method that only looks at significant \textit{p}-values and judges whether their distribution shows sign of being produced by studies with insufficient power. The \textit{p}-curve for estimation \citep{simonsohn2014} is a fixed effect selection model with a significance cut-off at $0.05$ estimated by minimizing the Kolmogorov-Smirnov distance \citep{mcshane2016adjusting}. Another method from the psychology literature is \textit{p}-uniform \citep{van2015meta}, which is similar to the \textit{p}-curve. A recent study by \citet{carter2019correcting} compared several approaches and showed that the selection model works better than the others. However, not even the best method works well in every considered scenario. For more information on publication bias and a good review of large part of these methods we refer to the book by \citet{rothstein2006publication}.

In contrast, \textit{p}-hacking, sometimes also called \emph{questionable research practices} \citep{Sijtsma2016} and \emph{fishing for significance} \citep{Boulesteix2009}, occurs when the authors of a study manipulate results into statistical significance. \textit{p}-hacking can be done at the experimental stage, using for example optional stopping, or at the analysis stage, for instance by changing models or dropping out participants. Examples of \textit{p}-hacking can be found in \citet{simmons2011false}. While publication bias is almost perfectly captured by selection models such as that of Hedges \citeyearpar{hedges1992modeling}, \textit{p}-hacking is much harder to model. The aforementioned \textit{p}-curve approach by \citet{simonsohn2014p} has been used for \textit{p}-hacking as well, but it has been shown to be not reliable \citep{BrunsIoannidis2016}. Here we advocate the selection model approach and propose to use it to model both publication bias and \textit{p}-hacking. We derive some properties for these models and argue they are best handled by Bayesian methods. 

The paper is organized as follows: In section \ref{sect:models} we define the framework and introduce the models, which are analysed and compared in section \ref{sect:differences}. Further comparisons are presented through simulations in section \ref{sect:simulations} and real data examples in section \ref{sect:examples}. We end with some concluding remarks in section \ref{sect:conclusions}.

\section{Models}\label{sect:models}
\subsection{Framework}
The main ingredient of a meta-analysis is a collection of exchangeable statistics $x_{i}$. Each statistic $x_{i}$ has density $f^{\star}(x_{i};\theta_{i},\eta_{i})$, where $\eta_i$ is a known or unknown nuisance parameter and $\theta_{i}$ is an unknown parameter we wish to do inference on. This paper is about the fact that the true data-generating model $f^{\star}(x_{i};\theta_{i},\eta_{i})$ is often not what it ideally should have been, such as a normal density. It has instead been transformed into something else by the forces of publication bias and \textit{p}-hacking. Our goal is to understand what it has been transformed into, and how we can estimate $\theta_{i}$ accordingly. The publication bias model of \citet{hedges1992modeling,iyengar1988selection} and the soon-to-be introduced \textit{p}-hacking model are models that transform the underlying densities, denoted by $f^{\star}(x_{i};\theta_{i},\eta_{i})$, into new densities, $f_{i}(x_{i};\theta_{i},\eta_{i})$. The underlying densities will usually be normal, but they do not have to. The theoretical discussion in this paper will not enforce normality anywhere, but all examples of models are based on underlying normal distributions. We only require the dependencies on a parameter of interest $\theta_{i}$ and that statistical inference on $\theta_{i}$ is the goal of the analysis.

The parameter $\theta_{i}$ is typically an effect size, such as a standardized mean difference. In a fixed effects meta-analysis, $\theta_{i}=\theta$ for all $i$. In a random effects meta-analysis, $\theta_{i}$ is drawn from an effect size distribution $p(\theta)$ common to all $i$, and the goal of the study is often to make inference on the parameters of the effect size distribution, for example on the mean $\theta_{0}$ and the standard deviation $\tau$ when $\theta \sim N(\theta_{0},\tau)$. If we marginalize away $\theta_{i}$ we will end up with a density on the form $f(x_{i}; \theta_{0},\sqrt{\sigma_{i}^{2}+\tau^{2}})$, assuming $x_{i}$ is also from a normal distribution with standard deviation $\sigma_{i}$ (i.e., $\eta_i = \sigma_i$). This is possible in our framework, but it turns out that an important property of the publication bias model gets lost, as marginalizing out the $\theta_{i}$s can mask the fact that the selection mechanism in the publication bias has an effect both on the effect size distribution and the individual densities $f_{i}(x_{i};\theta_i, \sigma_i)$.

In this paper we use Bayesian methods. While a frequentist approach is in theory possible, it will lead to poor results. As noted by \citet[Appendix, 1]{mcshane2016adjusting}, the one-sided random effects models have ridges in their likelihood, which may make non-regularized estimates imprecise. In particular, it can be proved \citep{Moss2019} there are no confidence sets of guaranteed finite size for $\theta_{0}$ and $\tau$ in the one-sided normal random effect models, for any coverage $1-\alpha$. This is problematic for two reasons: (i) It would be useless to report a confidence set for $\tau^{2}$ like $[0.5,\infty)$, as no one would be confident about an infinite value for that parameter; (ii) the automatic confidence sets procedures that are guaranteed to yield finite confidence set of some positive nominal coverage, such as bootstrapped confidence sets, likelihood-ratio based confidence sets, and subsampling confidence sets never have true coverage greater than $0$ \citep[see][]{gleser996bootstrap, Moss2019}. The role of priors in the Bayesian approach here is to force the estimates away from highly implausible areas; ad hoc penalization or bias corrections are necessary for frequentist methods to work well.

\subsection{The publication bias model} \label{subsect:publicationBias}

Imagine the publication bias scenario:
\begin{quote}
Alice is an editor who receives a study with a \textit{p}-value $u$. She knows her journals will suffer if she publishes many null-results, so she is disinclined to publish studies with large \textit{p}-values. Still, she will publish any result with some \textit{p}-value-dependent probability $w(u)$. Every study you will ever read in Alice's journal has survived this selection mechanism, the rest are lost forever.
\end{quote}
In this story, the underlying model $f^{\star}(x_{i}\mid\theta_{i},\eta_{i})$
is transformed into a publication bias model
\begin{equation}
f(x_{i}\mid\theta_{i},\eta_{i})\propto f^{\star}(x_{i}\mid\theta_{i},\eta_{i})w(u)\label{eq:Publication bias model}
\end{equation}
by the selection probability $w(u)$. Here $u$ is a \textit{p}-value that depends on $x_{i}$ and maybe something else, such as the standard deviation of $x_{i}$, but does not depend on $\theta_{i}$. It cannot depend on $\theta_{i}$ since the editor has no way of knowing the parameter $\theta_{i}$; if she did, she would not have to look at the \textit{p}-values at all. It might depend on other quantities modelled by $\eta_{i}$ though, if $\eta_{i}$ is known to the editor. The normalizing constant of model \eqref{eq:Publication bias model} is finite for any probability $w(u)$, hence $f$ is a \textit{bona fide} density.

An argument against the publication bias scenario is that publication bias does not act only through \textit{p}-values, but also through other features of the study such as language \citep{egger1998meta}
and originality \citep{callaham1998positive}. While this is true, the publication bias scenario seems to completely capture the idea of \textit{p}-value based publication bias. Moreover, the \textit{p}-value-based publication bias is more relevant to meta-analysis than the other sources of bias mentioned above. Even if other sources of publication bias exist, maybe acting through $x_{i}$ but not its \textit{p}-value, publication bias based on \textit{p}-values is a universally recognized problem, and a good place to start.

The kind of model sketched here is almost the same as the one of \citet{hedges1992modeling}, with the sole exception that \citet{hedges1992modeling} does not require $w(u)$ to be a probability. The only requirement is that the integral of $f^{\star}(x_{i}\mid\theta_{i},\eta_{i})w(u)$ is finite, which can happen without $w(u)$ being a probability). We demand that $w(u)$ to be a probability since the intuitive publication bias scenario interpretation of the model disappears when $w(u)$ is not a probability. Anyway, there are many choices for $w(u)$ even when we force it to be a probability. Assume $w^{\star}(u)$ is any bounded positive function in $\left[0,1\right]$, and define $w(u)=w^{\star}(u)/\textrm{sup}\{w^{\star}(u)\}$. Then $w(u)$ is a probability for each $u$, and fits right into the publication bias framework. An easy way to generate examples of such functions is to take density functions on $\left[0,1\right]$ and check if they are bounded. For instance, beta densities are bounded whenever both shape parameters are greater than $1$. The beta density is used in the publication bias model of \citet{citkowicz2017parsimonious}, but they do not demand it to be a probability.

Even if we know the underlying $f^{\star}(x_{i}\mid\theta_{i},\eta_{i})$ of model \eqref{eq:Publication bias model}, we will need to decide on what \textit{p}-value to use. Usually, the \textit{p}-value will be approximately a one-sided normal \textit{p}-value, but it might be something
else instead. A one-sided normal \textit{p}-value makes sense because most hypotheses have just one direction that is interesting. For instance, the effect of an antidepressant must be positive for the study to be publishable. A one-sided \textit{p}-value can also be used if the researchers reported a two-sided value, since $p=0.05$ for a two-sided hypothesis corresponds to $p=0.025$ for a one-sided hypothesis. We will use the one-sided normal \textit{p}-value in all examples in this paper.

Provided we know the underlying $f_{i}^{\star}$s and \textit{p}-values $u$, we only need to decide on the selection probability to have a fully specified model. \citet{hedges1992modeling} proposes the discrete selection probability
\begin{equation}
w(u\mid\rho,\alpha)=\sum_{j=1}^{J}\rho_{j}1_{(\alpha_{j-1},\alpha_{j}]}(u),\label{eq:Weighted model step function}
\end{equation}
where $\alpha$ is a vector with $0=\alpha_{0}<\alpha_{1}<\cdots<\alpha_{J}=1$ and $\rho$ is a non-negative vector with $\rho_{1}=1$. The interpretation of this selection probability is simple: When Alice reads the \textit{p}-value $u$, she finds the $j$ with $u\in(\alpha_{j-1},\alpha_{j}]$ and accepts the study with probability $\rho_{j}$. Related to this view, \citet{hedges1992modeling} proposed $\alpha_{[1,\dots,J-1]} = (0.001,0.005,0.01,0.05)$, as these \enquote{have particular salience for interpretation} \citep{hedges1992modeling}. In fact, a publication decision often depends on whether a \textit{p}-value crosses the $0.05$-threshold. His reason for using more split points than just $0.05$ is that \enquote{It is probably unreasonable to assume that much is known about the functional form of the weight function} \citep{hedges1992modeling}. While this is true, one may prefer, considering the bias-variance trade-off heuristic, to only use one split point at $0.05$, as done by \citet{iyengar1988selection} in their second weight function. Other reasons to prefer one split are ease of interpretation and presentation. Nevertheless, only using $0.05$ as a threshold for one-sided \textit{p}-values is problematic, as many published results are calculated using a two-sided \textit{p}-value instead. It seems therefore useful to add an additional splitting point at $0.025$, as a two-sided \textit{p}-value at that level corresponds to a one-sided \textit{p}-value of $0.05$. Ergo, we propose a two-step function selection probability
\[
w(u\mid\rho)=1_{[0,0.025)}(u)+\rho_{2}1_{[0.025,0.05)}(u)+\rho_{3}1_{\left[0.05,1\right]}(u),
\]
where the selection probability when $u\in[0,0.025)$ is normalized to $1$ to make the model identifiable.

The following proposition shows the densities of the one-sided normal step function selection probability publication bias models, with fixed effects and with normal random effects, respectively. Here the notation $\phi_{\alpha}(x;\theta,\sigma)$ indicates a normal truncated to $[a,b)$.
\begin{prop}
\label{prop:One-sided normal discrete probability vector publication bias model-1}
The density of an observation from a fixed effects one-sided normal step function selection probability publication bias model is
\begin{equation}\label{eq:Fixed effects, publication bias}
f(x_{i};\theta_{i},\sigma_{i}) = \sum_{j=1}^{N}\pi_{j}^\star\phi_{[\Phi^{-1}(1-\alpha_{j}),\Phi^{-1}(1-\alpha_{j-1}))}(x_{i}\mid\theta_{i},\sigma_{i}),
\end{equation}
where
$$
\pi_{j}^{\star}=\rho_{j}\frac{\Phi(c_{j-1}\mid\theta_{i},\sigma_{i})-\Phi(c_{j}\mid\theta_{i},\sigma_{i})}{\sum_{j=1}^{N}\rho_{j}\left[\Phi(c_{j-1}\mid\theta_{i},\sigma_{i})-\Phi(c_{j}\mid\theta_{i},\sigma_{i})\right]}
$$
and $c_{j}=\Phi^{-1}(1-\alpha_{j})$.

The density of an observation from the one-sided normal step function selection probability publication bias model with normal random effects and parameters $\sigma_{i},\theta_{0},\tau,$ is
\begin{equation}\label{eq:Random effects, publication bias}
f(x\mid\theta_{0},\tau,\sigma_{i})=\sum_{j=1}^{N}\pi_{j}^{\star}(\theta_0,\tau,\sigma_{i})\phi_{[\Phi^{-1}(1-\alpha_{j}),\Phi^{-1}(1-\alpha_{j-1}))}(x\mid\theta_{0},\sqrt{\tau^{2}+\sigma_{i}^{2}}),
\end{equation}
where 
\[
\pi_{j}^{\star}(\theta_0,\tau,\sigma_{i})=\rho_{j}\frac{\Phi(c_{j-1}\mid\theta_{0},\sqrt{\tau^{2}+\sigma_{i}^{2}})-\Phi(c_{j}\mid\theta_{0},\sqrt{\tau^{2}+\sigma_{i}^{2}})}{\sum_{j=1}^{J}\rho_{j}\left[\Phi(c_{j-1}\mid\theta_{0},\sqrt{\tau^{2}+\sigma_{i}^{2}})-\Phi(c_{j}\mid\theta_{0},\sqrt{\tau^{2}+\sigma_{i}^{2}})\right]}.
\]
\end{prop}

Here $f(x_i\mid\theta_{0},\tau,\sigma_i)$ is not equal to $\int f(x_{i};\theta_{i},\sigma_{i})\phi(\theta_{i};\theta_{0},\tau)d\theta_{i}$,
as it might have been expected. See the appendix for more details.

\subsection{The \textit{p}-hacking model}\label{subsect:p-hacking}
Imagine the \textit{p}-hacking scenario:
\begin{quote}
Bob is an astute researcher who is able to \textit{p}-hack any study to whatever level of significance he wishes. Whenever Bob does his research, he decides on a significance level to reach by drawing an $\alpha$ from a distribution $\omega$. Then he \textit{p}-hacks his study to this $\alpha$-level.
\end{quote}
In this scenario the original density $f^{\star}(x_{i};\theta_{i},\eta_{i}, u)$
is transformed into the \textit{p}-hacked density
\begin{equation}\label{eq:p-hacking model}
f(x_{i};\theta_{i},\eta_{i})=\int_{[0,1]}f_\alpha^{\star}(x_{i};\theta_{i},\eta_{i}, u)d\omega(\alpha),
\end{equation}
where $f_\alpha^{\star}$ is the density $f^{\star}$ truncated so that the \textit{p}-value $u\in\left[0,\alpha\right]$. Let us call the distribution $\omega$ the \emph{propensity to p-hack}. It might depend on covariates, but should not depend on $\theta_{i}$, as the researcher cannot know the true effect size of his study. While publication bias model (\ref{eq:Publication bias model}) is a selection model, the \textit{p}-hacking model (\ref{eq:p-hacking model}) is clearly a mixture model. The publication bias can also be written as a mixture model on the same form as the \textit{p}-hacking model, but then $\omega$ will depend on $\theta$, see the appendix. We stress the fact that the model \eqref{eq:p-hacking model} is not a publication bias model. Although the \textit{p}-hacking model can be written as a selection model (i.e., on the form of \eqref{eq:Publication bias model}), in general the publication probability will depend on the true effect size, which violates an obvious condition for a model to be considered a publication bias model. See section \ref{sect:differences} for a detailed comparison of the two models.

Just as the publication bias model requires a choice of $w$, the \textit{p}-hacking model requires a choice of $\omega$. A \textit{p}-hacking scientist is motivated to \textit{p}-hack to the $0.05$ level, maybe to the $0.01$ or $0.025$, but never to a level such as $0.07$ or $0.37$. This motivates the discrete \textit{p}-hacking probability distribution
$$\omega(\alpha\mid\pi)=\sum_{j=1}^{J}\pi_{j}1_{(0,\alpha_{j}]}(\alpha)$$
for some $j$-ary vector $\alpha$ satisfying $0<\alpha_{1}<\alpha_{2}<\cdots<\alpha_{J}=1$,
and $j$-ary vector of probabilities $\pi$. The resulting density is 
\[
f(x_{i};\theta_{i},\eta_{i})=\sum_{j=1}^{J}\pi_{j}\left(\int_{u\in(0,\alpha_{j}]}f^\star(x_{i};\theta_{i},\eta_{i}, u)d\omega(\alpha)\right)^{-1}f^\star(x_{i};\theta_{i},\eta_{i})1_{(0,\alpha_{j}]}(u).
\]
Using a reasoning entirely analogous to that of section \ref{subsect:publicationBias}, we define $\omega$ as
\[
\omega(u;\pi) = \pi_11_{[0,0.025]}(u) + \pi_{2}1_{(0,0.05]}(u) + \pi_{3}1_{(0,1]}(u),
\]
i.e., we only consider two splitting points at $0.025$ and $0.05$.

The density of an observation from a fixed effects one-sided normal discrete probability \textit{p}-hacking model is
\begin{eqnarray}
f(x_{i};\theta_{i},\sigma_{i}) & = & \sum_{j=1}^{J}\pi_{j}\phi_{[\Phi^{-1}(1-\alpha_{j}),\Phi^{-1}(1-\alpha_{j-1}))}(x_{i}\mid\theta_{i},\sigma_{i}),\label{eq:Fixed effects, p-hacking}
\end{eqnarray}
but there is no closed form for the density of its random effect version.

\section{The difference between the models}\label{sect:differences}

\subsection{Selection sets\label{sec:Selection Sets}}
There is a real but subtle difference between the publication bias model and the \textit{p}-hacking model. To properly understand this difference, let us introduce the idea of selection sets.

\begin{algorithm}[!h]
\begin{algorithmic}[1]
	\State $x^{0}\sim p(x)$.
	\For{$i$ in $i=0,1,\ldots$}
		\If{$s\mid x^i = 1$}         
			\State Report $x^i$.           
		\Else         
			\State $x_{H}^{i+1}\sim p(x_{H}^{i}\mid x_{H^{c}}^{0})$.        
		\EndIf  
	\EndFor  
\end{algorithmic}
\caption{\label{alg:Selection model}The selection model $q_{H}(x)$.}
\end{algorithm}

Let $X$ be a stochastic variable with density $p(x)$, such as a standardized effect size. Let the \emph{selection variable} $s$ be a binary stochastic variable that equals $1$ if and only if $X$ is observed. To understand the meaning of $s$, recall the publication bias scenario, where not all $X$s are observed, because they first have to be accepted by the editor. The variable $s$ equals $1$ if $X$ is accepted by the editor and $0$ otherwise. When $X$ is univariate, the density of our observed $X$ is $q(x)=p(s=1\mid x)/p(s=1)p(x)$. This is also known as a \emph{weighted distribution}, see e.g.\ \citep[][eq. 3.1]{rao1985weighted}. When $X$ is multivariate, we find ourselves in a slightly more difficult position. Now we have to state which variables to integrate over to recover the normalizing constant of $q(x)\propto p(s=1\mid x)p(x)$. Let us use the term \emph{selection set} for the set of variables to integrate over and denote the set of their indexes with capital letters, e.g., $H$. Making use of the notational convention $x_{H}=\left\{ x_{i},i\in H\right\}$, define the selection model based on $H$ as
\begin{equation}
q_{H}(x)=\frac{p(s=1\mid x)}{p(s=1\mid x_{H^{c}})}p(x)\label{eq:H-selection model},
\end{equation}
where $H^c$ is the complement of $H$. This model can be be viewed as a rejection sampling model \citep{von1951various}. To understand how it works, take a look at the pseudo-code in algorithm \ref{alg:Selection model}: $H$ is the set of every variable that is sampled together until $s=1$.

\begin{prop}\label{prop:is density}
The function $q_{H}(x)$ is a density for any $H$.
\end{prop}
See the appendix for a proof of Proposition \ref{prop:is density}. When $H$ contains every $x_i$, $q_H$ is the simplest kind of selection model, where every variable is sampled together until $s=1$. When $H$ is empty, no variables can be resampled, and the model reduces to $p(x)$. But for non-empty $H$, $q_{H}(x)$ will often be equal to neither $p(x)$ nor $q_{H\cup H^c}(x)$, and different choices of selection sets $H\neq G$ will usually lead to different models $q_{H}(x)\neq q_{G}(x)$. 
\begin{prop}
\label{prop:Equal selection models}Two selection models based on the same $p(x)$ and $s$ are equal, i.e. $q_{H}(x)=q_{G}(x)$, if and only if $p(s=1\mid x_{H^{c}})=p(s=1\mid x_{G^{c}})$.
In particular, $q_{H}(x)=p(x)$ if and only if $p(s=1\mid x_{H^{c}})=p(s=1\mid x)$.
\end{prop}

\begin{proof}
Both results follow directly from Equation \eqref{eq:H-selection model}.
\end{proof}

It is handy to visualize selection models and their selection sets using directed acyclic graphs. To this end recall that a Bayesian network is a directed acyclic graph $G$ together with a probability
density $p$ satisfying the property that $$p(x)=\prod_{v\in V(G)}p(x_{v}\mid x_{\textrm{pa}(v)}),$$ where $V(G)$ is the set of vertices in $G$ and $x_{\textrm{pa}(v)}$ are the parents of $x_{v}$ in $G$ \citep{Pearl2014}.
Transforming a Bayesian network for $p$ into a Bayesian network for $q_{H}$ is easy, just add the following to $G$: (i) The selection variable vertex $s$, and (ii) arrows $x$ to $s$ for each $x$ that $s$ depends on. Then
\begin{equation}
q_{H}(x\mid s=1)=\frac{p(s=1\mid x_{\textrm{pa}(s)})}{p(s=1\mid H^{c})}\prod_{v\in V(G)}p(x_{v}\mid x_{\textrm{pa}(v)})\label{eq:DAG, selection model}.
\end{equation}

To visualize the selection set $H$, start by drawing a dashed plate around the vertices in $H$. In plate notation \citep{buntine1994operations}, a solid plate represents variables that are sampled together. The
dashed plate does almost the same, for recall that $H$ contains all the elements that are sampled together until $s=1$. The semantic difference between a dashed and a solid plate is that every sample in a solid plate is observed, but only one of potentially many samples in a dashed plate is observed. The following bare-bones example should make things clear.
\begin{example}\label{exa:Marginal density of theta}
Let $p(x,\theta)=p(x\mid\theta)p(\theta)$
be a density and $s$ be a function of $x$ only, so that $p(s=1\mid x,\theta)=p(s=1\mid x)$.
The possible selection models are
\begin{eqnarray*}
q_{\emptyset}(x,\theta) = q_{\theta}(x,\theta) & = & p(x,\theta)\\
q_{(x,\theta)}(x,\theta) & = & \frac{p(s=1\mid x)}{p(s=1)}p(x,\theta)\\
q_{x}(x,\theta) & = & \frac{p(s=1\mid x)}{p(s=1\mid\theta)}p(x,\theta)
\end{eqnarray*}
Figure \ref{fig:Plate notation, simple example} displays the directed acyclic graphs of $p$ and the selection models when $H = \emptyset$, $H=\left\{ x,\theta\right\}$, and $H=\left\{ x\right\}$, respectively. The marginal distribution of $\theta$ is not the same for $H=\left\{ x\right\}$ and $H=\left\{ x,\theta\right\} $, as
\begin{eqnarray*}
q_{\left\{ x,\theta\right\} }(\theta) & = & p(\theta)\frac{p(s=1\mid\theta)}{p(s=1)}\\
q_{x}(\theta) & = & \int\frac{p(s=1\mid x)}{p(s=1\mid\theta)}p(x,\theta)dx=p(\theta),
\end{eqnarray*}
i.e., it is affected by the selection mechanism $s$.
\end{example}

\begin{figure}
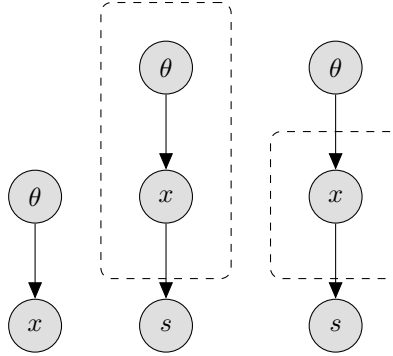

\begin{center}     
 \begin{tabular}{ccc}   
  \import{figures/}{simple_example_0} &   
   \import{figures/}{simple_example_1} & 
   \import{figures/}{simple_example_2}
 \end{tabular} 
\end{center}

\caption{\label{fig:Plate notation, simple example} Three simple selection models. {\bf (left)} the original $p(x,\theta)$; {\bf (middle)} the model $q_{\left\{ x,\theta\right\} }$, where $\theta$ and $x$ are
sampled together until $s=1$; {\bf (right)} the model $q_{x}$, where only $x$ is sampled until $s=1$.}
\end{figure}

Let $f(x,\theta\mid\eta)=f(x\mid\theta,\eta)p(\theta)$
be the joint density of a random effects meta-analysis, where $x$ is the effect size, $\theta$ is the study-specific parameter of interest, and $\eta$ is a study-specific nuisance parameter such as the sample size of the study. The left plot of Figure \ref{fig:Plate notation, simple example} is a visualization of $f(x,\theta)$. If we have more than one study to analyse, we will have to work with product density $\prod_{i=1}^{n}f(x_{i},\theta_{i}\mid\eta_{i})$ instead of the stand-alone density $f(x,\theta\mid\eta)$. This is visualised in the middle plot of Figure \ref{fig:Plate notation, simple example}  by drawing a solid plate around the pair $(x,\theta)$. When we are dealing with a fixed effects meta-analysis, in which $\theta$ is fixed, the plate should be drawn around $x$ only (Figure \ref{fig:Plate notation, simple example}, right graph).

\subsection{Publication bias and \textit{p}-hacking models\label{subsec:Selection sets, meta analysis}}

To visualize selection models based on \textit{p}-values, we must make some modifications to the original graph:

\begin{enumerate}[label=\roman*]
\item Add the \textit{p}-value node $u$.
\item Add an arrow from $x$ to $u$. 
\item Since the \textit{p}-value $u$ usually depends on more information than just $x$, such as the standard deviation of $x$, add an arrow from $\eta$ (which represents the extra information) to $u$ as well.
\item Add the selection node $s$ and an arrow from $u$ to $s$. If $u$ is the only parent of $s$, we are dealing with selection models only based on
\textit{p}-values. 
\end{enumerate}

The placement of dashed and solid plates depends on which model we want to use.

The idea behind the publication bias model is that a completely new study is done whenever the last one failed to be published. This implies that $\theta$ and $x$ are sampled together. The left plot of Figure \ref{fig:Plate notation, publication bias and p-hacking} shows the direct acyclic graph of the normal publication bias model defined in Proposition \ref{prop:One-sided normal discrete probability vector publication bias model-1}. In this particular case, $\eta$ corresponds to $\sigma$, the standard deviation. Moreover, $u$ is a \textit{p}-value, $\theta_{0}$ is the mean of the effect size distribution, $\tau$ is the standard deviation of the effect size distribution, and $\rho$ is the selection probability function. The variable $Z$ lives on the unit interval, and encodes the editor's decision to publish: If the observed \textit{p}-value is less than $Z$, the study is published. Importantly, $Z$ is placed inside the selection set because a new \textit{p}-value cut-off decision is made for each study received. Since $x$ and $\theta$ are sampled together, the selection mechanism modifies $p(\theta)$.

In the \textit{p}-hacking scenario, the \textit{p}-hacker will hack his study all the way to significance, regardless of $\theta$. This means that $\theta$ and $x$ are sampled separately and $\theta$ must be placed outside the selection set. Moreover, the decision of how much to \textit{p}-hack is not reevaluated at each attempt. Consequently, the random variable that controls the \textit{p}-hacking decisions, analogously to the publication bias model, $Z$, is also placed outside the selection graph. This is the case, for example, of an author who decides to \textit{p}-hack to level $\alpha$ ($Z = \alpha$): he acts on $x$ to obtain the desired \textit{p}-value, whatever the sampled $\theta$ is. The graphical representation of this model is shown in the right plot of Figure \ref{fig:Plate notation, publication bias and p-hacking}. Since $x$ and $\theta$ are not sampled together, the selection mechanism does not modify $p(\theta)$.

\begin{figure}
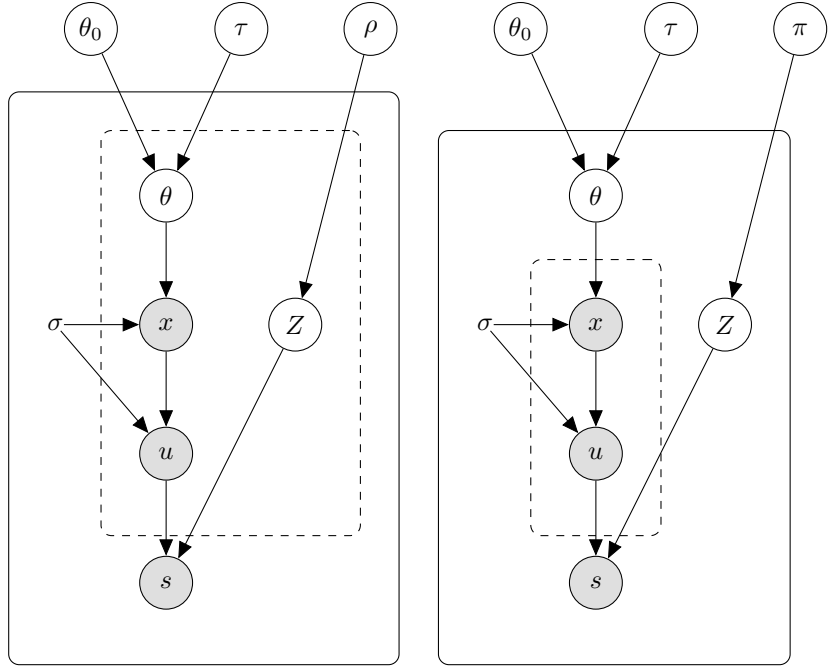

\begin{center}     
 \begin{tabular}{cc}    
  \import{figures/}{publication_bias_type_i} &     
  \import{figures/}{phacking_type_ii}
 \end{tabular} 
\end{center}
\caption{\label{fig:Plate notation, publication bias and p-hacking} Directed acyclic graphs for: {\bf (left)}
the publication bias model; {\bf (right)} the \textit{p}-hacking model. The dashed plates enclose the selection sets and the the solid plates enclose variables that are repeated together.}
\end{figure}

\subsection{Equivalence in special cases}
The publication bias model defined in Proposition \ref{prop:One-sided normal discrete probability vector publication bias model-1} and the \textit{p}-hacking model are equivalent when $\sigma_{i}$ is fixed across studies. This holds both for the fixed and random effects models. To see this, let $\pi$ be any probability vector for the \textit{p}-hacking model and solve the invertible linear system $\pi^{\star}(\rho)=\pi$ for $\rho$. There is no guarantee for the models to be equivalent when $\sigma_{i}$ is not fixed, see the appendix.

\section{Simulations}\label{sect:simulations}

We want to answer these three questions about the \textit{p}-hacking and publication bias models:
\begin{enumerate}[label=\roman*]
\item Do they work even in the absence of \textit{p}-hacking and publication bias? Although we know these phenomena are ubiquitous and should always be corrected for, it is still important that the models do not distort the results when there is no publication bias or \textit{p}-hacking.
\item How do they behave in extreme situations? In particular, we are interested in how the models behave when $n$ is small and the heterogeneity is large.
\item Are the models distinguishable in practice? Does the \textit{p}-hacking model work under the publication bias scenario and vice versa?
\end{enumerate}

\subsection{Settings}
We generate data under three scenarios: (i) With no publication bias nor \textit{p}-hacking, using the normal random effect meta-analysis model. (ii) Under the presence of publication bias, using model \eqref{eq:Random effects, publication bias}. (iii) Under presence of \textit{p}-hacking, using the random effects normal \textit{p}-hacking model. The study-specific variances $\sigma_{i}^{2}$ are sampled uniformly from $\left\{ 20,\ldots80\right\} $. The size of the meta-analyses are $n = 5, 30, 100$, corresponding to small, medium and large meta-analyses, while the means for the effect size distribution are $0, 0.2, 0.8$. The value $\theta_0 = 0$ corresponds to no expected effect, while the positive $\theta_0$s are the cut-off for small and large effect sizes of \citet[][pages 24 -- 27]{cohen1988statistical}. The standard deviations of the random effects distributions are $\tau=0.1$ and $\tau=0.5$. While $\tau = 0.1$ is a reasonable amount of heterogeneity, $\tau=0.5$ is a large amount of heterogeneity that provides a challenge for the models. The probability of acceptance of a paper are simulated to be $1$ if the \textit{p}-value is between $0$ and $0.025$, $0.7$ if the \textit{p}-value is between $0.025$ and $0.05$, and $0.1$ otherwise. For the same intervals, the \textit{p}-hacking probabilities are $0.6$, $0.3$ and $0.1$. 

For each parameter combination we estimate the \textit{p}-hacking model and the publication bias model. Both models have normal likelihoods and normal effect size distributions. We use one-sided significance cut-offs at $0.025$ and $0.05$ for both the publication bias and the \textit{p}-hacking models. We use standard Gaussian priors for $\theta_0$, a standard half normal prior for $\tau$, and, in the \textit{p}-hacking model, a uniform Dirichlet prior for $\pi$. For the $\rho$ in the publication bias model we use a a uniform Dirichlet that constrains $\rho_{1}\geq\rho_{2}\geq\ldots\geq\rho_{j}$. That is, the publication probability is a decreasing function of the \textit{p}-value.

All of these priors are reasonable. A standard normal for $\theta_0$ is reasonable because we know that $\theta_0$ has a small magnitude in pretty much any meta-analysis, and most are clustered around $0$. A half normal prior for $\tau$ is also reasonable, as $\tau$ is much more likely to be very small than very big. The priors for $\rho$ and $\pi$ are harder to reason about, but a uniform Dirichlet seems like a natural and neutral choice. These are the standard prior of the $\mathtt{R}$ package $\mathtt{publipha}$ \citep{publipha}, which we used for all computations. $\mathtt{publipha}$ uses $\mathtt{STAN}$ \citep{Carpenter2017-cf} to estimate the models, and each estimation uses $8$ chains.

The number of simulations is $N = 100$ for each parameter combination. See the OSF repository for this paper (\url{https://osf.io/tx8qn/}) for the code used to run the simulations and its raw data.

\subsection{Results}
\paragraph{No publication bias, no \textit{p}-hacking.} The results under this scenario are reported in Table \ref{tab:Simulation_classical}. When the amount of heterogeneity is reasonable ($\tau = 0.1$) both the \textit{p}-hacking and the publication bias perform well. The publication bias model performs slightly worse than the \textit{p}-hacking model when the mean effect size is large ($\theta_0 = 0.8$) and the number of studies small ($n=5$), but it catches up as $n$ increases. The situation with $\tau = 0.5$ is more interesting. Here the \textit{p}-hacking model outperforms the publication bias model, with the latter tending to underestimate the mean effect. While increasing $n$ alleviates the problem, there is still a substantial underestimation of $\theta_0$ even in the case of $n = 100$. In contrast, both models seems to estimate $\tau$ pretty well. As a take home message, it is safe to use the \textit{p}-hacking model when there is no \textit{p}-hacking or publication bias, but less safe to use the publication bias model.

\begin{table}
\noindent
\caption{\label{tab:Simulation_classical} {\bf No publication bias, no \textit{p}-hacking.} Posterior means and standard deviations from the \textit{p}-hacking and publication bias models when the data are simulated from the normal random effects meta-analysis model.}
\begin{center}
\begin{tabular}{llllrrrrrrrc}
\multicolumn{3}{r}{\textbf{True values}} &  & \multicolumn{3}{c}{\textbf{\textit{p}-hacking model}} &  & \multicolumn{3}{c}{\textbf{Publication bias model}} & \tabularnewline
$\tau$ & $\theta_0$ & $n$ &  & \multicolumn{1}{c}{$\widehat{\theta_0}$} &  & \multicolumn{1}{c}{$\widehat{\tau}$} &  & \multicolumn{1}{c}{$\widehat{\theta_0}$} &  & \multicolumn{1}{c}{$\widehat{\tau}$} & \tabularnewline
\hline
\multirow{9}{*}{$0.1$} & \multirow{3}{*}{$0$} & $5$ &  & -0.03 (0.09) &  & 0.18 (0.07) &  & -0.06 (0.08) &  & 0.13 (0.06) & \tabularnewline
 &  & $30$ &  & -0.01 (0.03) &  & 0.08 (0.03) &  & -0.02 (0.03) &  & 0.07 (0.03) & \tabularnewline
 &  & $100$ &  & -0.01 (0.02) &  & 0.08 (0.03) &  & -0.01 (0.02) &  & 0.07 (0.02) & \tabularnewline
 \cdashline{3-11}
 & \multirow{3}{*}{$0.2$} & $5$ &  &  0.12 (0.08) &  & 0.21 (0.08) &  &  0.09 (0.07) &  & 0.17 (0.08) & \tabularnewline
 &  & $30$ &  &  0.17 (0.04) &  & 0.09 (0.04) &  &  0.15 (0.03) &  & 0.09 (0.04) & \tabularnewline
 &  & $100$ &  &  0.18 (0.02) &  & 0.09 (0.03) &  &  0.17 (0.02) &  & 0.09 (0.03) & \tabularnewline
 \cdashline{3-11}
 & \multirow{3}{*}{$0.8$} & $5$ &  &  0.78 (0.08) &  & 0.21 (0.10) &  &  0.63 (0.15) &  & 0.34 (0.14) & \tabularnewline
 &  & $30$ &  &  0.80 (0.04) &  & 0.11 (0.04) &  &  0.80 (0.04) &  & 0.11 (0.04) & \tabularnewline
 &  & $100$ &  &  0.80 (0.02) &  & 0.10 (0.03) &  &  0.80 (0.02) &  & 0.10 (0.03) & \tabularnewline
 \cline{2-12}
\multirow{9}{*}{$0.5$} & \multirow{3}{*}{$0$} & $5$ &  & -0.03 (0.20) &  & 0.59 (0.21) &  & -0.21 (0.17) &  & 0.53 (0.21) & \tabularnewline
 &  & $30$ &  & -0.03 (0.09) &  & 0.51 (0.08) &  & -0.14 (0.09) &  & 0.47 (0.08) & \tabularnewline
 &  & $100$ &  & -0.02 (0.05) &  & 0.50 (0.04) &  & -0.08 (0.06) &  & 0.48 (0.04) & \tabularnewline
 \cdashline{3-11}
 & \multirow{3}{*}{$0.2$} & $5$ &  &  0.10 (0.22) &  & 0.57 (0.20) &  & -0.09 (0.19) &  & 0.54 (0.19) & \tabularnewline
 &  & $30$ &  &  0.15 (0.10) &  & 0.53 (0.08) &  &  0.02 (0.10) &  & 0.51 (0.08) & \tabularnewline
 &  & $100$ &  &  0.19 (0.05) &  & 0.51 (0.04) &  &  0.11 (0.06) &  & 0.49 (0.04) & \tabularnewline
 \cdashline{3-11}
 & \multirow{3}{*}{$0.8$} & $5$ &  &  0.68 (0.23) &  & 0.62 (0.21) &  &  0.35 (0.23) &  & 0.74 (0.21) & \tabularnewline
 &  & $30$ &  &  0.78 (0.10) &  & 0.52 (0.08) &  &  0.60 (0.14) &  & 0.60 (0.08) & \tabularnewline
 &  & $100$ &  &  0.79 (0.05) &  & 0.51 (0.04) &  &  0.70 (0.07) &  & 0.55 (0.04) & \tabularnewline
 \hline
\end{tabular}
\end{center}
\end{table}

\paragraph{Publication bias.} 
Overall, the publication bias model outperforms the \textit{p}-hacking model when the data are generated from the publication bias model, but not by much, see Table \ref{tab:Simulation_pb}). When $\tau = 0.5$ the \textit{p}-hacking model tends to overestimates $\theta_0$ while the publication bias model tends to underestimate it, and the overestimation of the \textit{p}-hacking model is most extreme when $\theta_0 = 0.2$. When $\tau = 0.1$, the models produce almost indistinguishable results. In the most challenging case of $n=5$, the \textit{p}-hacking model is never worse than the publication bias model and is closer to the truth when $\theta_0 = 0.8$. Just as in the \textit{p}-hacking scenario, both models estimate $\tau$ reasonably well.

\begin{table}
\noindent
\caption{\label{tab:Simulation_pb} {\bf Publication bias.} Posterior means and standard deviations from the \textit{p}-hacking and publication bias models when the data are simulated from the publication bias model with cut-offs at $0.025$ and $0.05$, with selection probabilities equal to $1$, $0.7$, and $0.1$ in the intervals $[0, 0.025)$, $[0.025, 0.05)$, and $[0.5, 1]$.}
\begin{center}
\begin{tabular}{llllrrrrrrrc}
\multicolumn{3}{r}{\textbf{True values}} &  & \multicolumn{3}{c}{\textbf{\textit{p}-hacking model}} &  & \multicolumn{3}{c}{\textbf{Publication bias model}} & \tabularnewline
$\tau$ & $\theta_0$ & $n$ &  & $\widehat{\theta_0}$ &  & $\widehat{\tau}$ &  & $\widehat{\theta_0}$ &  & $\widehat{\tau}$ & \tabularnewline
\hline
\multirow{9}{*}{$0.1$} & \multirow{3}{*}{$0$} & $5$ &  & -0.01 (0.10) &  & 0.23 (0.08) &  & -0.01 (0.07) &  & 0.18 (0.07) & \tabularnewline
 &  & $30$ &  &  0.02 (0.04) &  & 0.12 (0.05) &  &  0.01 (0.04) &  & 0.10 (0.04) & \tabularnewline
 &  & $100$ &  &  0.02 (0.03) &  & 0.12 (0.03) &  &  0.00 (0.02) &  & 0.10 (0.03) & \tabularnewline
 \cdashline{3-11}
 & \multirow{3}{*}{$0.2$} & $5$ &  &  0.10 (0.15) &  & 0.30 (0.09) &  &  0.10 (0.07) &  & 0.21 (0.08) & \tabularnewline
 &  & $30$ &  &  0.22 (0.05) &  & 0.11 (0.05) &  &  0.19 (0.05) &  & 0.09 (0.04) & \tabularnewline
 &  & $100$ &  &  0.23 (0.03) &  & 0.10 (0.04) &  &  0.20 (0.04) &  & 0.09 (0.03) & \tabularnewline
 \cdashline{3-11}
 & \multirow{3}{*}{$0.8$} & $5$ &  &  0.77 (0.08) &  & 0.20 (0.08) &  &  0.62 (0.14) &  & 0.32 (0.12) & \tabularnewline
 &  & $30$ &  &  0.80 (0.03) &  & 0.10 (0.04) &  &  0.79 (0.03) &  & 0.10 (0.04) & \tabularnewline
 &  & $100$ &  &  0.80 (0.02) &  & 0.10 (0.02) &  &  0.80 (0.02) &  & 0.10 (0.02) & \tabularnewline
 \cline{2-12}
 \multirow{9}{*}{$0.5$} & \multirow{3}{*}{$0$} & $5$ &  &  0.34 (0.21) &  & 0.53 (0.20) &  &  0.04 (0.22) &  & 0.56 (0.18) & \tabularnewline
 &  & $30$ &  &  0.36 (0.10) &  & 0.48 (0.09) &  &  0.01 (0.19) &  & 0.50 (0.08) & \tabularnewline
 &  & $100$ &  &  0.36 (0.04) &  & 0.47 (0.04) &  & -0.01 (0.10) &  & 0.50 (0.04) & \tabularnewline
 \cdashline{3-11}
 & \multirow{3}{*}{$0.2$} & $5$ &  &  0.42 (0.21) &  & 0.54 (0.22) &  &  0.12 (0.22) &  & 0.59 (0.19) & \tabularnewline
 &  & $30$ &  &  0.50 (0.07) &  & 0.44 (0.08) &  &  0.16 (0.18) &  & 0.51 (0.09) & \tabularnewline
 &  & $100$ &  &  0.51 (0.04) &  & 0.42 (0.04) &  &  0.19 (0.10) &  & 0.50 (0.05) & \tabularnewline
 \cdashline{3-11}
 & \multirow{3}{*}{$0.8$} & $5$ &  &  0.81 (0.22) &  & 0.56 (0.19) &  &  0.47 (0.27) &  & 0.71 (0.20) & \tabularnewline
 &  & $30$ &  &  0.90 (0.09) &  & 0.45 (0.08) &  &  0.64 (0.21) &  & 0.58 (0.13) & \tabularnewline
 &  & $100$ &  &  0.90 (0.04) &  & 0.45 (0.04) &  &  0.74 (0.09) &  & 0.53 (0.06) & \tabularnewline
\hline
\end{tabular}
\end{center}
\end{table}

\paragraph{\textit{p}-hacking.} The simulation results for the \textit{p}-hacking model are in Table \ref{tab:Simulation_ph}. As before, the largest differences are in the most difficult case of $\tau = 0.5$ while the two models tend to agree in the more realistic case of $\tau = 0.1$. When $\tau = 0.5$ the publication bias model severely underestimates $\theta_0$, even getting the sign wrong in some instances. This should not come as a surprise given the interpretation of $\theta_0$ in the publication bias model, but demonstrates we should be cautious in interpreting the $\theta_0$ estimates. 

\begin{table}
\caption{\label{tab:Simulation_ph} {\bf \textit{p}-hacking.} Posterior means and standard deviations from the \textit{p}-hacking and publication bias models when the data are simulated  from the \textit{p}-hacking model with cut-offs at $0.025$ and $0.05$, with \textit{p}-hacking probabilities equal to $0.6$, $0.3$, and $0.1$ in the intervals $[0, 0.025)$, $[0.025, 0.05)$, and $[0.5, 1]$}
\begin{center}
\begin{tabular}{llllrrrrrrrc}
\multicolumn{3}{r}{\textbf{True values}} &  & \multicolumn{3}{c}{\textbf{\textit{p}-hacking model}} &  & \multicolumn{3}{c}{\textbf{Publication bias model}} & \tabularnewline
$\tau$ & $\theta_0$ & $n$ &  & $\widehat{\theta_0}$ &  & $\widehat{\tau}$ &  & $\widehat{\theta_0}$ &  & $\widehat{\tau}$ & \tabularnewline
 \hline
 \multirow{9}{*}{$0.1$} & \multirow{3}{*}{$0$} & $5$ &  & -0.06 (0.14) &  & 0.29 (0.07) &  &   0.04 (0.06) &  & 0.17 (0.05) & \tabularnewline
 &  & $30$ &  & -0.02 (0.08) &  & 0.13 (0.05) &  &   0.01 (0.07) &  & 0.07 (0.03) & \tabularnewline
 &  & $100$ &  &  0.00 (0.05) &  & 0.10 (0.04) &  &   0.00 (0.05) &  & 0.05 (0.02) & \tabularnewline
 \cdashline{3-11}
 & \multirow{3}{*}{$0.2$} & $5$ &  &  0.12 (0.16) &  & 0.29 (0.09) &  &   0.10 (0.06) &  & 0.21 (0.06) & \tabularnewline
 &  & $30$ &  &  0.18 (0.06) &  & 0.12 (0.05) &  &   0.15 (0.06) &  & 0.09 (0.03) & \tabularnewline
 &  & $100$ &  &  0.20 (0.04) &  & 0.09 (0.04) &  &   0.17 (0.05) &  & 0.08 (0.03) & \tabularnewline
 \cdashline{3-11}
 & \multirow{3}{*}{$0.8$} & $5$ &  &  0.79 (0.08) &  & 0.18 (0.09) &  &   0.65 (0.14) &  & 0.30 (0.13) & \tabularnewline
 &  & $30$ &  &  0.80 (0.03) &  & 0.10 (0.04) &  &   0.79 (0.03) &  & 0.10 (0.04) & \tabularnewline
 &  & $100$ &  &  0.80 (0.02) &  & 0.10 (0.02) &  &   0.80 (0.02) &  & 0.10 (0.02) & \tabularnewline
 \cline{2-12}
 \multirow{9}{*}{$0.5$} & \multirow{3}{*}{$0$} & $5$ &  &  0.08 (0.22) &  & 0.47 (0.19) &  &   0.01 (0.12) &  & 0.37 (0.19) & \tabularnewline
 &  & $30$ &  &  0.08 (0.09) &  & 0.43 (0.08) &  &  -0.24 (0.19) &  & 0.35 (0.10) & \tabularnewline
 &  & $100$ &  &  0.07 (0.06) &  & 0.44 (0.04) &  &  -0.33 (0.14) &  & 0.37 (0.06) & \tabularnewline
 \cdashline{3-11}
 & \multirow{3}{*}{$0.2$} & $5$ &  &  0.19 (0.24) &  & 0.50 (0.20) &  &   0.05 (0.13) &  & 0.42 (0.22) & \tabularnewline
 &  & $30$ &  &  0.24 (0.09) &  & 0.47 (0.08) &  &  -0.20 (0.19) &  & 0.46 (0.09) & \tabularnewline
 &  & $100$ &  &  0.23 (0.05) &  & 0.47 (0.04) &  &  -0.27 (0.16) &  & 0.47 (0.06) & \tabularnewline
 \cdashline{3-11}
 & \multirow{3}{*}{$0.8$} & $5$ &  &  0.72 (0.19) &  & 0.60 (0.19) &  &   0.35 (0.20) &  & 0.73 (0.19) & \tabularnewline
 &  & $30$ &  &  0.78 (0.09) &  & 0.52 (0.07) &  &   0.36 (0.23) &  & 0.67 (0.11) & \tabularnewline
 &  & $100$ &  &  0.80 (0.05) &  & 0.50 (0.04) &  &   0.42 (0.20) &  & 0.65 (0.09) & \tabularnewline
\hline
\end{tabular}
\end{center}
\end{table}

\section{Examples}\label{sect:examples}
In this section we apply the models on the two meta-analyses of \citet{cuddy2018p} and \citet{anderson2010violent}. As in the simulation study, we use normal models for each effect size with one-sided significance cut-off at $0.025$ and $0.05$ for both models. We use the same priors as we did in the simulation study. To compare the fit of the models we use the leave-one-out cross-validation information criterion (\textsc{LOOIC}) \citep{loo_article}, calculated using the R \citep{R} package $\mathtt{loo}$ \citep{loo}. LOOIC equals $-2\cdot\textsc{elpd}_{\textsc{loo}}$, where $\textsc{elpd}$ is the expected log pointwise predictive density for a new data set and $\textsc{elpd}_\textsc{loo}$ is an estimate of this quantity by leave-one-out cross validation. Just as the \textsc{AIC} \citep{akaike1998information}, smaller values indicate better model fit. As for the simulation study, the analyses have been done with the $\mathtt{R}$ package $\mathtt{publipha}$ \citep{publipha}, which in turn uses \texttt{STAN} \citep{Carpenter2017-cf}. Each model has been estimated with $8$ chains. See the OSF repository for this paper (\url{https://osf.io/tx8qn/}) for the code used to run the examples.

\subsection{Power posing\label{subsec:cuddy2018}}

\citet{cuddy2018p} conducted a meta-analysis of a of power posing, an alleged phenomenon where adopting expansive postures has positive psychological feedback effects. Their meta-analysis is not conventional, but a \textit{p}-curve analysis \citep{simonsohn2014p}. A \textit{p}-curve analysis is not based on estimated effect sizes and standard errors, but directly on \textit{p}-values. The data from \citep{cuddy2018p} can be accessed via the Open Science Framework (\url{https://osf.io/pfh6r/}). Here we only consider studies with outcome \enquote{mean difference}, design \enquote{2 cell}, and test statistic that is either $F$ or $t$. The $F$-statistics are all with $1$ denominator degree of freedom, and the root of these are distributed as the absolute value of a $t$-distributed variable. The $t$-values and the roots of the $F$-statistics are converted to standardized mean differences by using $d = t\sqrt{2/\nu}$, where $\nu$ is the degrees of freedom for the $t$-test. The standardized mean differences are to the left i Figure \ref{fig:cuddy2017}. Note the outlier $x_{12} = 1.72$. As it has a large effect on all the models, we analyze the data both with and without $x_{12}$.

\begin{figure}
\noindent \begin{centering}
\includegraphics[width=0.49\textwidth]{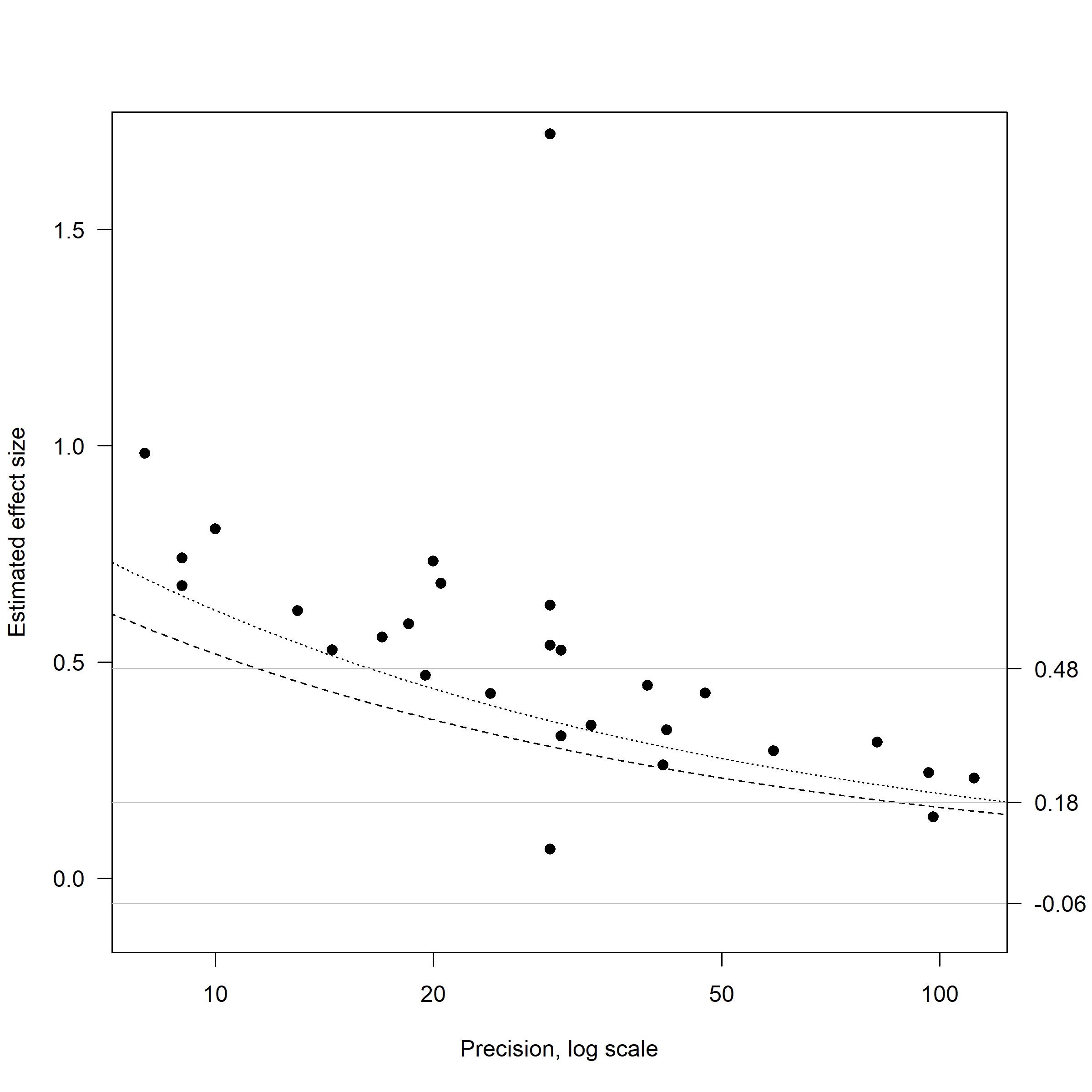}\includegraphics[width=0.49\textwidth]{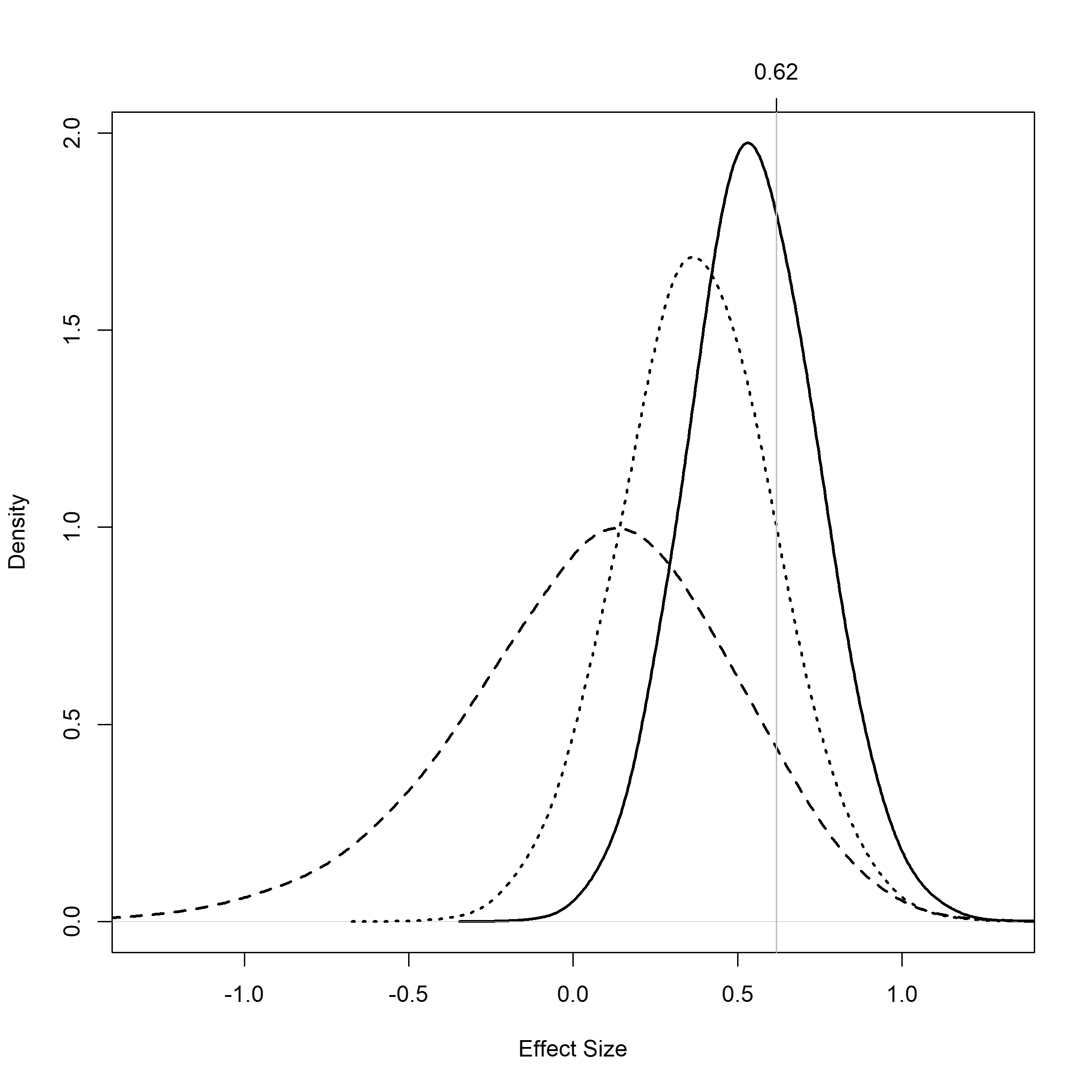}
\par\end{centering}
\caption{\label{fig:cuddy2017} \textbf{(left)} Effect sizes for the power posing example. The dotted black line is $1.96/\textrm{sd}$ and the dashed black line is $1.64/\textrm{sd}$. The ticks on the right hand side are
the meta-analytic means: $0.48$ is from the uncorrected model, $0.17$ is the mean of the selected effect size distribution under the \textit{p}-hacking model, while $-0.06$ is the mean under the publication bias model. \textbf{(right)} Posterior densities for $\theta_{2}$ in the power posing example. The dashed density belongs to the \textit{p}-hacking model, the dotted density to the publication bias model, and the solid density to the uncorrected model. The point $x_{2}=0.62$ is marked for reference.}
\end{figure}

The estimates of the \textit{p}-hacking model, the publication bias model, and the uncorrected meta-analysis models are in Table \ref{tab:cuddy2018}. According to the LOOIC the corrected models account much better for the data than the uncorrected model. Both the \textit{p}-hacking model and the publication bias models estimate
larger $\tau$s and smaller $\theta_{0}$s than the classical model, with the publication bias model estimating the surprising $\theta_{0}\approx0$. But recall the results of the simulation study, where the publication bias model severely underestimates $\theta_0$ when the \textit{p}-hacking model is true.

The publication bias selection affects not only the observed $x_{i}$s, but also the $\theta_{i}$s. As a consequence, the posterior mean of the selected effect size distribution (this equals $0.37$, is not shown in the table, and equals the average of the posterior means for the $\theta_{i}$s) is much closer to the uncorrected model's estimate than the \textit{p}-hacked estimate. This effect can be most easily understood by looking at a specific $\theta$, for example the $\theta_2$ reported in the right plot of Figure \ref{fig:cuddy2017}, where $x_{2}=0.62$. In this case, the publication bias posterior for is close to the uncorrected posterior even though $\theta_0 \approx 0$. On the other hand, the \textit{p}-hacking model pushes $0.62$ down to $0.17$, towards the meta-analytic mean of $.18$.

Finally, the surprisingly low value for $\theta_0$ obtained with the publication bias model can be a side effect of the presence of the outlier $x_{12} = 1.72$. Its presence on the right tail of an hypothetical true effect size distribution implies unobserved low and negative effects not reported due to publication bias. When the outlier is removed from the analysis, the estimate of $\theta_{0}$ goes up and agrees with the estimate from the \textit{p}-hacking model, which does not change. Once the outlier is removed, the fit of the publication bias model increases tremenduously, reaching a level close to that of the \textit{p}-hacking model. Moreover, the estimates of $\tau$ are strongly affected by the removal of $x_{12}$. In particular, the estimate of $\tau$ decreases from $0.45$ to $0.09$ in the \textit{p}-hacking model.

\begin{table}
\caption{\label{tab:cuddy2018} Power posing example: Posterior means for LOOICs and parameters (mean effect $\theta$, standard deviation $\tau$, probabilities of \emph{p}-hacking $\pi$/probabilities of being published $\rho$) of the \emph{p}-hacking, publication bias, and classical meta-analysis (uncorrected) model estimated on the data by \citet{cuddy2018p}. The results in the top table are obtained with all studies, those in the bottom without the outlier $x_{12}$. Posterior standard deviations are reported between brackets.}
\noindent
\begin{center}%
\begin{tabular}{lrrrrrrrrr}
\multicolumn{10}{c}{\bf All studies}\\
  & LOOIC & \multicolumn{1}{c}{$\theta_{0}$} & \multicolumn{1}{c}{$\tau$} & \multicolumn{1}{c}{$\pi_{1}/\rho_{1}$} & \multicolumn{1}{c}{$\pi_{2}/\rho_{2}$} \\ 
  \hline
uncorrected & 16 (18) & 0.48 (0.07) & 0.27 (0.06) &   &   \\ 
  \emph{p}-hacking & -18 (14) & 0.18 (0.12) & 0.45 (0.10) & 0.62 (0.15) & 0.23 (0.14) \\ 
  publication bias & -5.1 (22) & -0.06 (0.23) & 0.37 (0.09) & 0.39 (0.22) & 0.03 (0.03) \\ 
   \hline
\tabularnewline
\multicolumn{10}{c}{\bf Without outlier}\\
  & LOOIC & \multicolumn{1}{c}{$\theta_{0}$} & \multicolumn{1}{c}{$\tau$} & \multicolumn{1}{c}{$\pi_{1}/\rho_{1}$} & \multicolumn{1}{c}{$\pi_{2}/\rho_{2}$} \\ 
  \hline
uncorrected & -7.1 (5.7) & 0.39 (0.04) & 0.09 (0.05) &   &   \\ 
  \emph{p}-hacking & -38 (10) & 0.18 (0.07) & 0.09 (0.07) & 0.62 (0.15) & 0.24 (0.15) \\ 
  publication bias & -35 (11) & 0.16 (0.09) & 0.08 (0.06) & 0.26 (0.17) & 0.03 (0.03) \\ 
   \hline
\end{tabular}
\end{center}
\end{table}


In conclusion, the \textit{p}-hacking and publication bias models suggest there is selection bias in these studies. Both models have much better fit than the uncorrected one and it is reasonable to accept their parameter estimates as more realistic. Nontheless, both models agree on a value of $\theta_{0}$ that is likely to be different from $0$. The results of Table \ref{tab:cuddy2018} supports \citet{cuddy2018p}'s conclusion that there is evidence for some positive effect of power posing. The \textit{p}-hacking model does not suffer the presence of an outlier, and, in contrast to the publication bias model, provides similar results with and without $x_{12}$ in the data.

\subsection{Violent video games\label{subsec:Anderson}}

\citet{anderson2010violent} conducted a large meta-analysis on the effects of violent video games on seven negative outcomes such as aggressive behavior and aggressive cognition. As part of their analysis, they classified some experiments as best practice experiments \citep[for more details, see Table 2 of][]{anderson2010violent}. Suspecting publication bias, \citet{hilgard2017overstated} reanalysed the data using an array of tools to detect and adjust for publication bias. For the outcome variable aggressive cognition, \citet{hilgard2017overstated} noted that \enquote{Application of best-practices criteria seems to emphasize statistical significance, and a knot of experiments just reach statistical significance}. The data can be found on the web \citep{Hilgard2017} and are visualised to the left in Figure \ref{fig:anderson2010}. In the plot, the best practice experiments are represented by solid circles, all other experiments by hollow squares. An outlier $x=1.33$ has been removed from the data set, and excluded from our analyses. Its removal substantially improves the fit for all the models. 

In this example we fit the three models (\textit{p}-hacking, publication bias and uncorrected models) to three data subsets (all experiments, only best practice experiments, without best practice experiments). The outcome variable is aggressive behavior. Our the aim is to answer the following:
\begin{enumerate}
\item What are the parameter estimates, in each subset, for each model?
\item Which model has the best fit?
\item Do we have a reason to believe the best practice experiments are drawn from a different underlying distribution than the other experiments, as \citet{hilgard2017overstated} and the top left plot of Figure \ref{fig:anderson2010} suggest?
\item Is there a large difference between the posterior for $\theta_{0}$ and the mean posterior for the $\theta_{i}$s, as we saw in the previous example?
\end{enumerate}

\begin{table}
\noindent
\caption{\label{tab:Anderson2010}Violent video games example: Posterior means for LOOICs and parameters (mean effect $\theta$, standard deviation $\tau$, probabilities of \emph{p}-hacking $\pi$/probabilities of being published $\rho$) of the \emph{p}-hacking, publication bias, and classical meta-analysis (uncorrected) model estimated on the aggressive behavior data from \citet{anderson2010violent}. Posterior standard deviations are reported between brackets.}
\begin{center}%
\begin{tabular}{lcccccccccc}
\multicolumn{11}{c}{\bf All Experiments} \tabularnewline
  & LOOIC & \multicolumn{1}{c}{$\theta_{0}$} & \multicolumn{1}{c}{$\tau$} & \multicolumn{1}{c}{$\pi_{1}/\rho_{1}$} & \multicolumn{1}{c}{$\pi_{2}/\rho_{2}$} \\ 
  \hline
uncorrected & -38 (11) & 0.18 (0.02) & 0.04 (0.03) &   &   \\ 
  \emph{p}-hacking & -48 (13) & 0.09 (0.04) & 0.05 (0.04) & 0.25 (0.11) & 0.23 (0.11) \\ 
  publication bias & -54 (13) & 0.08 (0.03) & 0.03 (0.02) & 0.44 (0.18) & 0.13 (0.07) \\ 
   \hline
\tabularnewline
\multicolumn{11}{c}{\bf Only Best Practice Experiments} \tabularnewline
  & LOOIC & \multicolumn{1}{c}{$\theta_{0}$} & \multicolumn{1}{c}{$\tau$} & \multicolumn{1}{c}{$\pi_{1}/\rho_{1}$} & \multicolumn{1}{c}{$\pi_{2}/\rho_{2}$} \\ 
  \hline
uncorrected & -42 (6.2) & 0.22 (0.02) & 0.03 (0.02) &   &   \\ 
  \emph{p}-hacking & -59 (12) & 0.10 (0.05) & 0.06 (0.04) & 0.37 (0.17) & 0.41 (0.17) \\ 
  publication bias & -61 (11) & 0.11 (0.04) & 0.03 (0.02) & 0.46 (0.21) & 0.06 (0.05) \\ 
   \hline
\tabularnewline
\multicolumn{11}{c}{\bf Without Best Practice Experiments} \tabularnewline
  & LOOIC & \multicolumn{1}{c}{$\theta_{0}$} & \multicolumn{1}{c}{$\tau$} & \multicolumn{1}{c}{$\pi_{1}/\rho_{1}$} & \multicolumn{1}{c}{$\pi_{2}/\rho_{2}$} \\ 
  \hline
uncorrected & -7.4 (5.7) & 0.06 (0.04) & 0.08 (0.05) &   &   \\ 
  \emph{p}-hacking & -6.2 (5.1) & 0.01 (0.05) & 0.07 (0.05) & 0.10 (0.07) & 0.11 (0.08) \\ 
  publication bias & -7.7 (5) & 0.02 (0.04) & 0.06 (0.04) & 0.61 (0.23) & 0.35 (0.19) \\ 
   \hline
\end{tabular}
\end{center}
\end{table}


The first three questions can be answered by looking at Table \ref{tab:Anderson2010}. The estimates of $\theta_{0}$ are approximately the same for the publication bias and \textit{p}-hacking models, and roughly half of the uncorrected estimate in all cases. In particular, when all experiments or only the best experiments are considered, there is a noticeable difference. In these two cases, the LOOICs suggest that some \textit{p}-hacking or publication bias is present, as they are smaller than the LOOIC for the uncorrected models. Although the publication bias model seems to work slightly better than the \textit{p}-hacking model, we can state that the two models agree and we have little reason to prefer one to the other. Basically, we can interpret this as converging evidence that the parameter estimates obtained with these two models for $\theta_{0}$ and $\tau$ are in the ballpark of their true values.

Interestingly, when we exclude the experiments not considered best practice by \citet{anderson2010violent}, the differences between the estimates provided by the corrected and uncorrected models reduce and the LOOICs are almost the same. The question is if the differences between best practice and non-best practice studies reflect a different underlying distribution or not. To answer this question, let us take a look at the posterior densities for $\theta_{0}$ when all experiments are included, as reported in the top right plot of Figure \ref{fig:anderson2010}. In this case, the posterior distributions computed with the \textit{p}-hacking and publication bias models are similar (dashed and dotted lines, respectively), which strengthens the agreement seen in Table \ref{tab:Anderson2010}. There is no large difference between the posterior for $\theta_{0}$ and the mean posterior for the $\theta_{i}$s as in the previous example. The answer to question (4) is therefore no.

Back to question (3). We have good reasons to believe the best practice experiments have been drawn from a different underlying distribution than the other experiments if there is negligible overlap between the posteriors for the parameters $\theta_{0}$. The uncorrected model supports this hypothesis (bottom right plot of Figure \ref{fig:anderson2010}), but the \textit{p}-hacking and publication bias models to do not. See the bottom left plot of Figure \ref{fig:anderson2010} for the posteriors for $\theta_0$ in the publication bias model (those obtained with the \textit{p}-hacking model are indistinguishable). In this case, the overlap between the posteriors for the different subsets is not negligible, and there is no evidence against hypotheses of equal $\theta_0$s in both groups. The same conclusion can be reached from Table \ref{tab:Anderson2010} by looking at the posterior standard deviations and posterior means.

\begin{figure}
\includegraphics[width=0.49\textwidth]{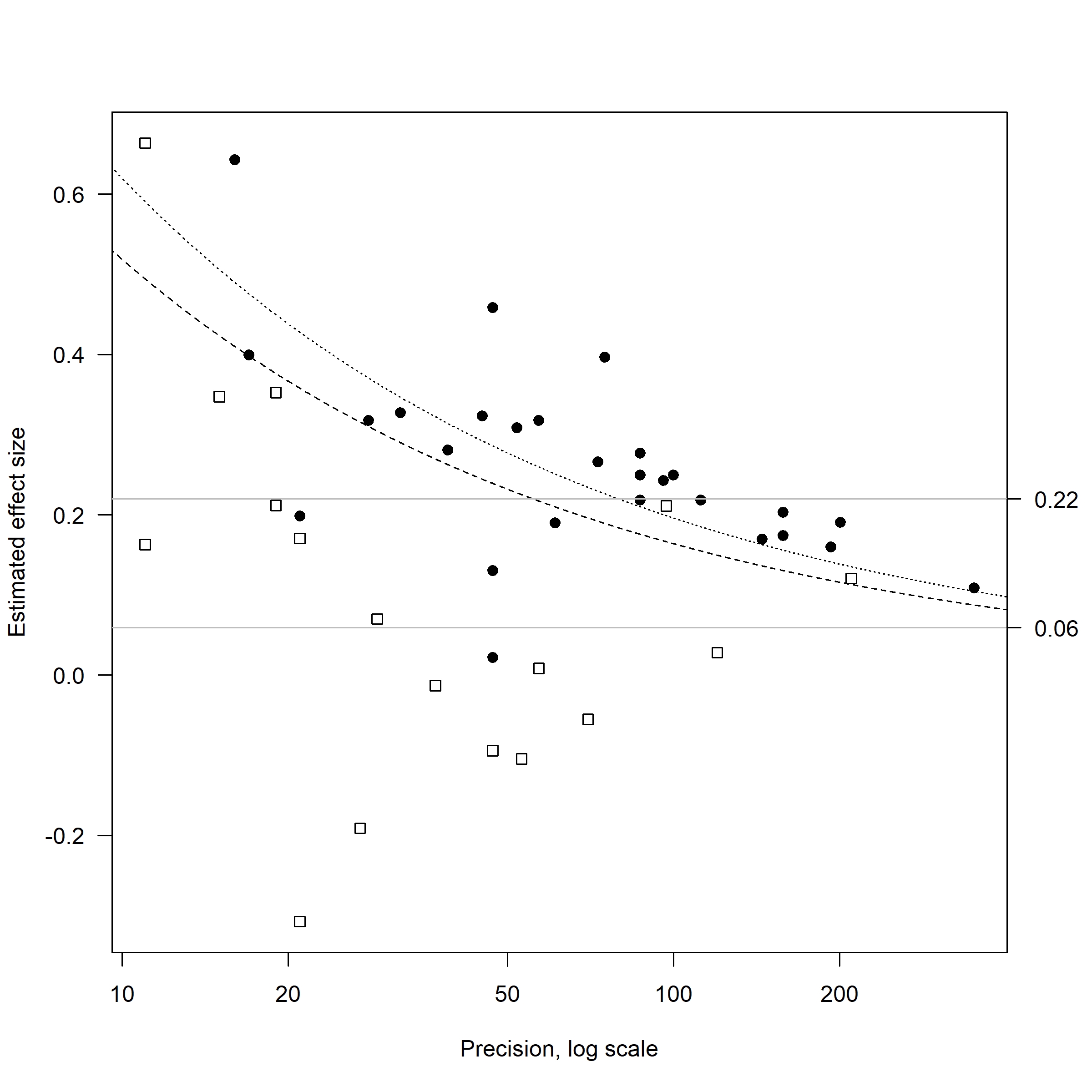}
\includegraphics[width=0.49\textwidth]{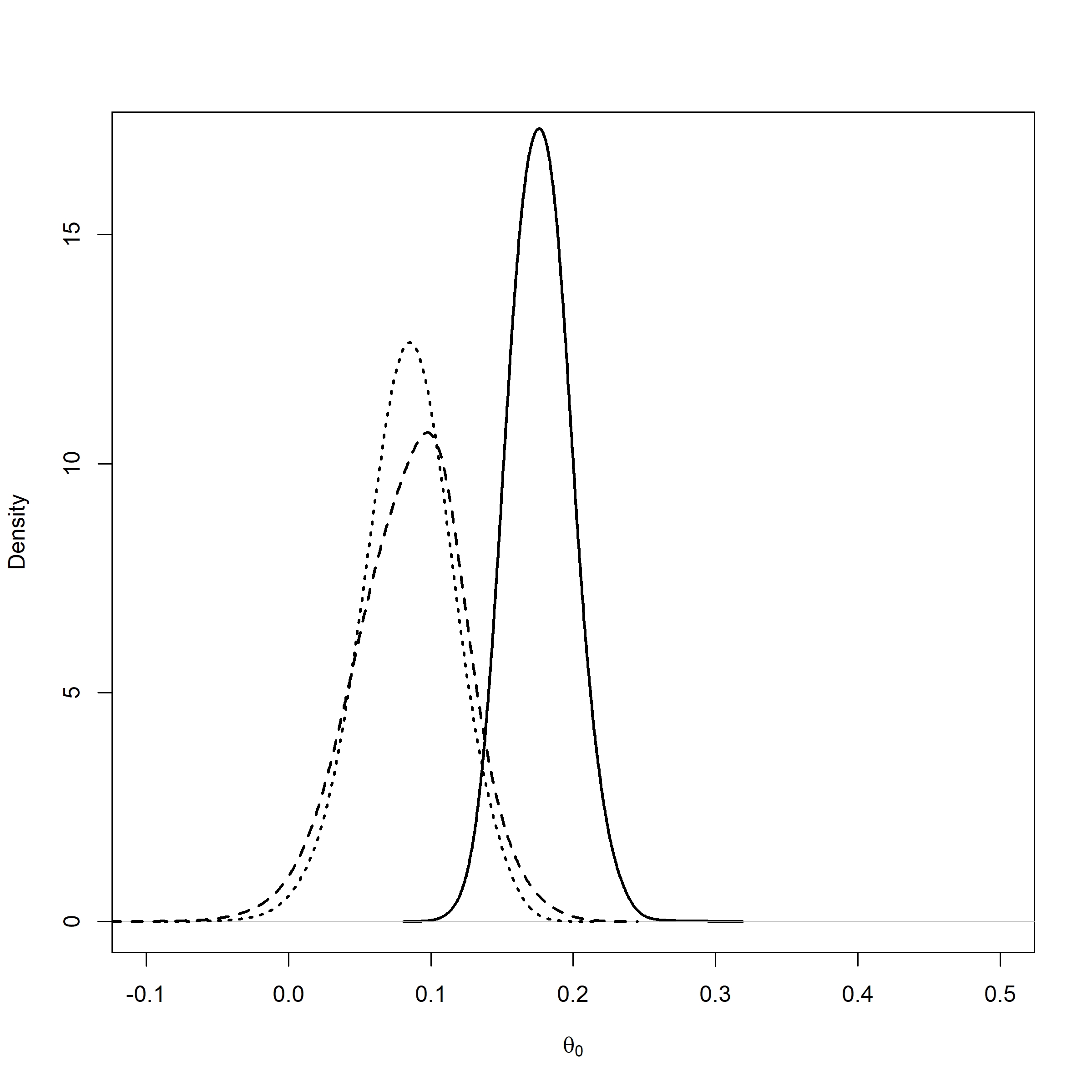}
\includegraphics[width=0.49\textwidth]{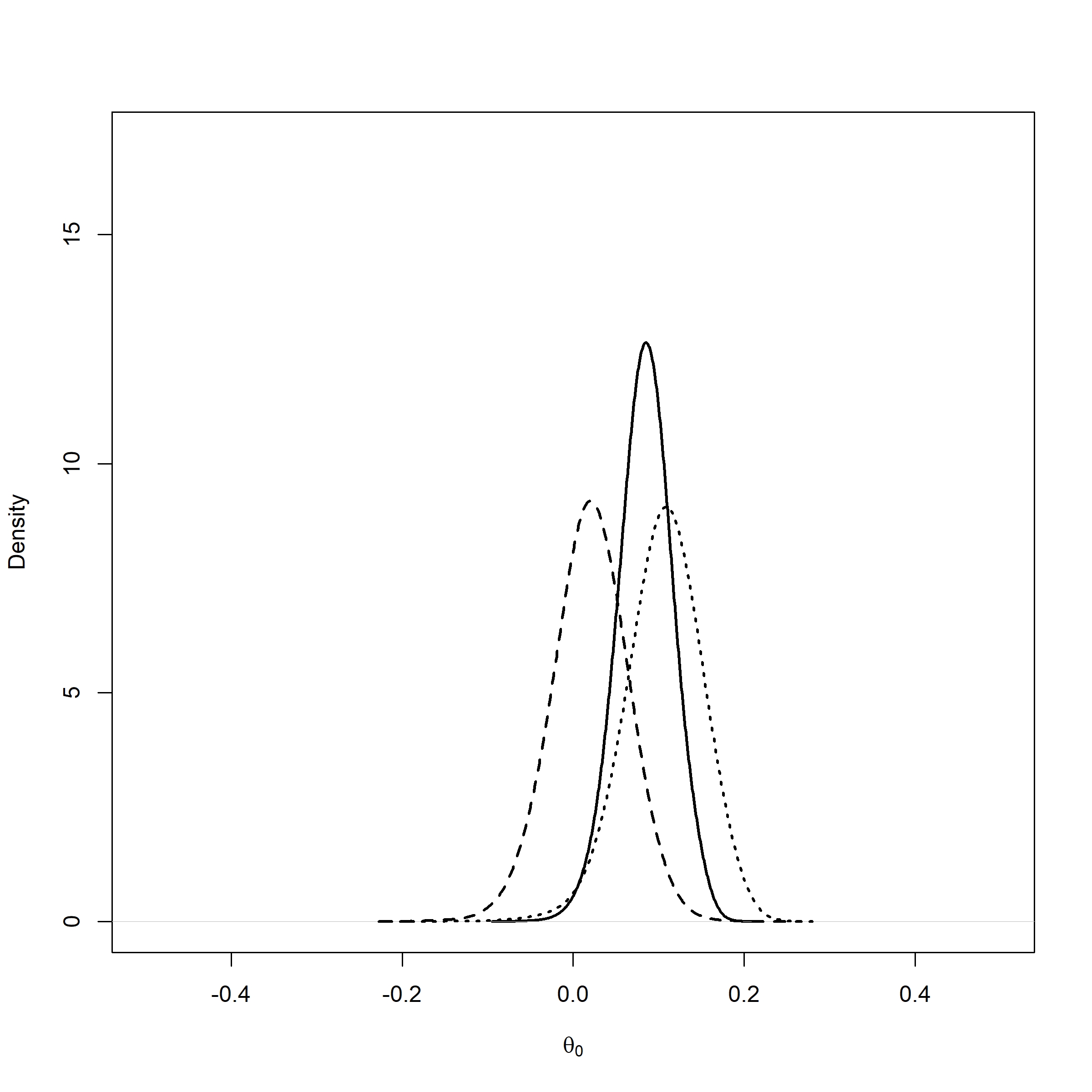}
\includegraphics[width=0.49\textwidth]{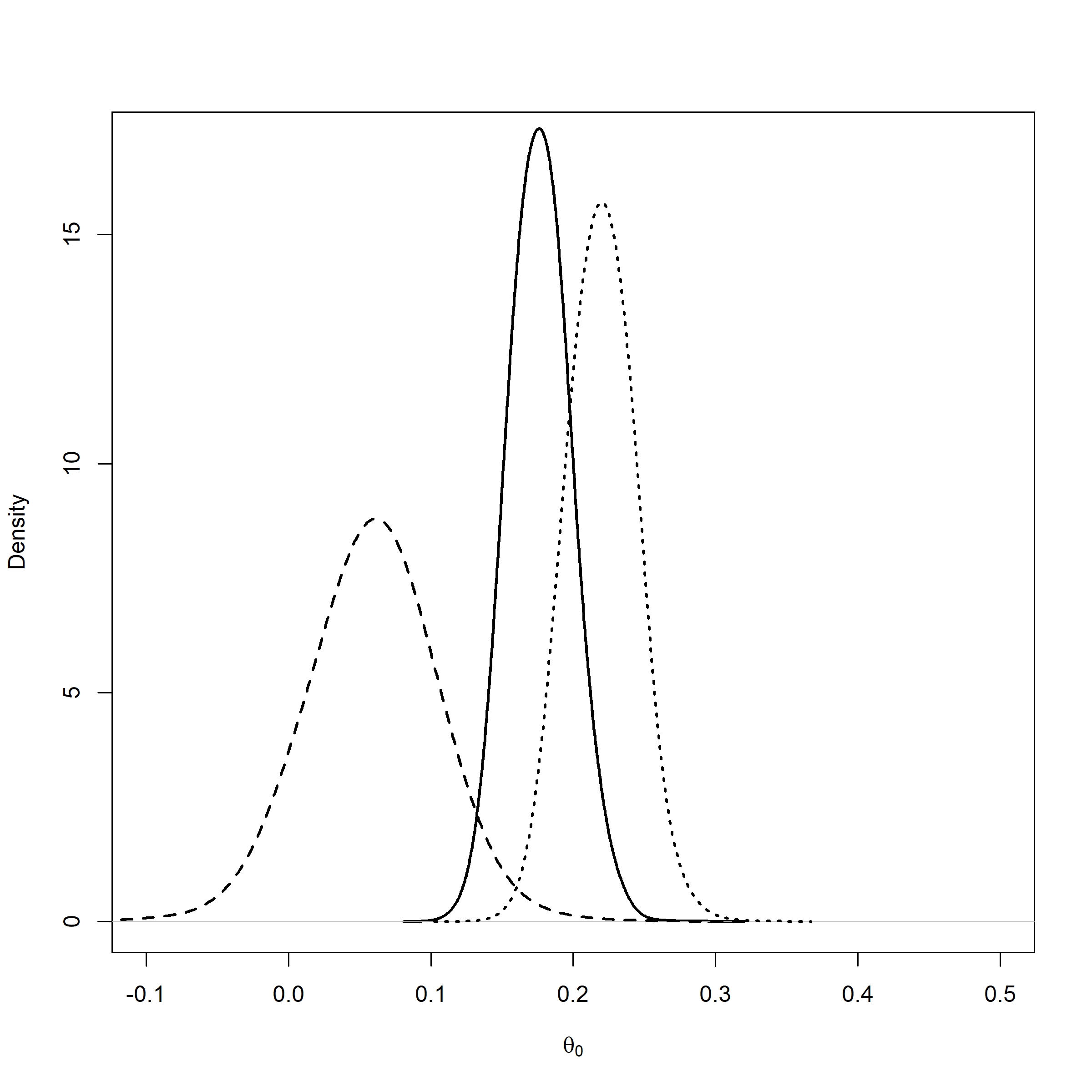}
\caption{\label{fig:anderson2010}Violent video games example with outcome variable aggressive behavior. \textbf{(top-left)} Effect sizes. The dotted black line is $1.96/\textrm{sd}$ and the dashed black line is $1.64/\textrm{sd}$. The ticks on the right hand side are the uncorrected meta-analytical means for each group: $0.29$ for the best practices group, $0.08$ for the rest. The outlier $x=1.33$ has been removed from the plot.
\textbf{(top-right)} Posterior densities for $\theta_{0}$ with all experiments included. The dashed density belongs to the \textit{p}-hacking model, the dotted to the publication bias model, and the solid to the uncorrected model. \textbf{(bottom-left)} Posterior densities for $\theta_{0}$ from the publication bias model. The solid curve is the model with all experiments, the dotted curve the model with the best practice experiments, and the dashed line the model without the best experiments. The posteriors for the \textit{p}-hacking model are similar to this one. \textbf{(bottom-right)} Posterior densities for $\theta_{0}$ (solid line: all experiments; dotted line: best practice experiments only; and dashed line without the best experiments) from the uncorrected meta-analysis model.}
\end{figure}

\section{Concluding remarks}\label{sect:conclusions}

In this paper we studied two models to handle the effect of \textit{p}-hacking and publication bias. Although the \textit{p}-hacking model worked really well in the simulation study, we have to admit that the \textit{p}-hacking scenario described in section \ref{subsect:p-hacking} is less plausible than the publication bias scenario of section \ref{subsect:publicationBias}. First, the assumption of Bob's \textit{p}-hacking omnipotence is strong. For while some researchers are able \textit{p}-hackers, most give up at some point. Does truncation actually model \textit{p}-hacking in the wild? Analysing \textit{p}-hacking is hard without serious simplifying assumptions. The model we proposed is interpretable and implementable, and it appears to work well in practice, as one can see in the examples of section \ref{sect:examples}. That said, there is space for further development of models for \textit{p}-hacking.

Regarding possible further development, we are often interested in understanding and modelling the sources of heterogeneity in a meta-analysis \citep{thompson1994systematic}. A way to do this is to let $\theta_{i}$ linearly depend on covariates, in the meta-analysis context known as moderators. If we extend the one-sided discrete models publication bias and \textit{p}-hacking models to include covariates, we will be able to estimate their effect while keeping the \textit{p}-hacking probability or the selection probability fixed. Another option is to allow the \textit{p}-hacking probability or the selection probability to depend on covariates themselves. For instance, the difficulty of \textit{p}-hacking is likely to increase with $n$, the sample size of the study. Similarly, the selection probability is also likely to be influenced by $n$; for example when $n$ is large, null-effects are more publishable.

The notation introduced in section \ref{sec:Selection Sets} can be used to visualize modifications of the two concrete models used in this paper, as in Figure \ref{fig:Plate notation, publication bias and p-hacking}. In the publication bias model, the nuisance parameter $\eta$ (which can include, e.g., the standard deviation $\sigma$) could be put inside the selection plate. In this case, new $\eta$s are drawn until a study is accepted. A possible modification of the \textit{p}-hacking model consists in putting $\theta$ inside the selection set, which makes the researcher draw new $\theta$s every time he attempts a \textit{p}-hack. This could be used to model scenarios where the hypothesis is not
known in advance by the researchers.

We saw in the simulations and in Example \ref{subsec:cuddy2018} that the publication bias and the \textit{p}-hacking models can give remarkably different results even with similar priors and the same $\alpha$ vector. A way to react to this situation is to choose the best-fitting model in terms of, for example, LOOIC. Nevertheless, this may result dangerous, and one should be caution, to not risk to over-interpret the results. More safely, one can present the results of both models and try to understand the differences between them, as we did in the examples of section \ref{sect:examples}. In the publication bias model, it is especially important to be aware of the interpretation of $\theta_{0}$ as the mean of the underlying effect size distribution, not the effect size distribution of the observed studies. Therefore, the best response to the question \enquote{Should one use the \textit{p}-hacking and publication bias model?} is probably \enquote{Use both!}

Finally, it would be interesting to model publication bias and \textit{p}-hacking at the same time:
\begin{quote}
Bob \textit{p}-hacks his research to a \textit{p}-value drawn from $\omega$ and sends it to Alice's journal. Alice accepts the paper with probability $w(u)$. Every rejected study is lost.
\end{quote}
In this scenario the original density $f^{\star}(x_{i}\mid\theta_{i},\eta_{i})$ is transformed twice: First by \textit{p}-hacking, then by publication bias. The resulting model is $$f(x_{i}\mid\theta_{i},\eta_{i})\propto w(u)\int_{[0,1]}f_{[0,\alpha]}^{\star}(x_{i}\mid\theta_{i},\eta_{i})d\omega(\alpha).$$ This is a reasonable model, but its normalizing constant is hard to calculate, even when $\omega$ is discrete and $w$ is a step function. Additional work on this problem is required.

\section*{Appendix}
\begin{proof}[Proof of Proposition \ref{prop:is density}]
\label{proof:is density}
We only need to show that $q_{H}(x)$ integrates to $1$. \begin{eqnarray*}
\int q_{H}(x)dx & = & \int\frac{p(s=1)}{p(s=1\mid x_{H^{c}})}p(x)dx\\
 & = & \int\frac{p(s=1\mid x_{H},x_{H^{c}})}{p(s=1\mid x_{H^{c}})}p(x_{H^{c}}\mid x_{H})p(x_{H})dx_{H^{c}}dx_{H}\\
 & = & \int\frac{p(s=1\mid x_{H^{c}})}{p(s=1\mid x_{H^{c}})})p(x_{H^{c}})dx_{H^{c}})\\
 & = & 1
\end{eqnarray*}
\end{proof}
As mentioned in section \ref{sect:differences}, any \textit{p}-hacking model can be written on the form of a selection model. Observe that
\begin{eqnarray*}
\int_{[0,1]}f_\alpha^{\star}(x_{i}\mid\theta_{i},\eta_{i}, u)d\omega(\alpha) & = & \int_{[0,1]}f(x_{i}\mid\theta_{i})P(u\in\left[0,\alpha\right]\mid\theta_{i},\eta_{i})^{-1}d\omega(\alpha)\\
 & = & f(x_{i}\mid\theta_{i})\int_{[0,u]}P(u\in\left[0,\alpha\right]\mid\theta_{i},\eta_{i})^{-1}d\omega(\alpha).
\end{eqnarray*}
where $f_\alpha^{\star}$ is the density $f^{\star}$ truncated so that the \textit{p}-value $u\in\left[0,\alpha\right]$. This is a publication bias model if $$h(u)=\int_{[0,u]}P(u\in\left[0,\alpha\right]\mid\theta_{i},\eta_{i})^{-1}d\omega(\alpha)$$ is bounded for each $u$ and $h(u)$ is independent of $\theta_{i},\eta_{i}$. While $h(u)$ can be bounded, it is typically dependent of $\theta_{i},\eta_{i}$, with the fixed effect model under complete selection for significance being a notable exception.

On the other hand, any selection model $f(x_{i};\theta_{i},\eta_{i})\rho(u)$ with $I =\int f(x;\theta_{i},\eta_{i})\rho(u)du<\infty$ can be written as a mixture model. For then there is a finite measure $d\omega(\alpha;\theta_{i},\eta_{i})$ satisfying 
\[
\rho(u)=\int_{[0,u]}\frac{1}{P(u\in\left[0,\alpha\right]\mid\theta,\eta)}d\omega(\alpha;\theta_{i},\eta_{i})
\]
Just take $d\omega(\alpha;\theta_{i},\eta_{i})=d\rho(\alpha)P(u\in\left[0,\alpha\right]\mid\theta_{i},\eta_{i})$, where $d\rho(\alpha)$ is defined by $\int_{0}^{u}d\rho(\alpha)=\rho(u)$. The size of the measure is
\begin{eqnarray*}
\int_{0}^{1}d\omega(\alpha;\theta_{i},\eta_{i}) & = & \int_{0}^{1}P(u\in\left[0,\alpha\right]\mid\theta_{i},\eta_{i})d\rho(\alpha)\\
 & = & \int_{0}^{1}f(u;\theta_{i},\eta_{i})\int_{0}^{u}d\rho(\alpha)du\\
 & = & I
\end{eqnarray*}
Hence $I_{\theta,\eta}d\omega'(\alpha;\theta_{i},\eta_{i})$ is a probability measure. This probability measure makes 
\[
I^{-1}f(x_{i};\theta_{i},\eta_{i})\rho(u)=\int_{[0,1]}f_\alpha(x_{i};\theta_{i},\eta_{i})d\omega'(\alpha)
\]
as can be seen by the following computation,
\begin{eqnarray*}
I^{-1}f(x_{i};\theta_{i},\eta_{i})\rho(u) & = & I^{-1}\int_{[0,u]}\frac{f(x_{i};\theta_{i},\eta_{i})}{P(u\in\left[0,\alpha\right]\mid\theta_{i},\eta_{i})}d\omega(\alpha)\\
 & = & I^{-1}\int_{[0,1]}\frac{f(x_{i};\theta,\eta)1_{\left[0,\alpha\right]}(u)}{P(u\in\left[0,\alpha\right]\mid\theta_{i},\eta_{i})}d\omega(\alpha)\\
 & = & I^{-1}\int_{[0,1]}f_\alpha(x_{i};\theta_{i},\eta_{i})d\omega(\alpha)\\
 & = & \int_{[0,1]}f_\alpha(x;\theta_{i},\eta_{i})d\omega'(\alpha)
\end{eqnarray*}

Proposition \ref{prop:One-sided normal discrete probability vector publication bias model-1} shows the form of the one-sided normal step function selection probability publication bias model when it is written as a mixture model of the form \eqref{eq:p-hacking model}. But most such mixture models are not true \textit{p}-hacking models, as the mixing probabilities $\pi_{i}^{\star}$ depend on $\theta$. There is no way for the \textit{p}-hacker to know
$\theta$, so we cannot regard the publication bias model as a \textit{p}-hacking model.

\bibliographystyle{biom}
\bibliography{main.bib}

\label{lastpage}
\end{document}